\documentclass{article}

\usepackage{graphicx}
\usepackage{tikz}
\usetikzlibrary{arrows, automata, patterns}
\usepackage{lipsum}
\usepackage{enumerate}
\usepackage{makecell}
\usepackage[unicode=true]{hyperref}

\usepackage[linesnumbered,ruled,vlined,boxed,algo2e]{algorithm2e}
\usepackage{nicefrac}
\usepackage{amsmath}
\usepackage{float}
\usepackage{mathtools}
\usepackage{amsthm}
\usepackage{multirow}
\usepackage{fullpage}
\usepackage{amsfonts}
\usepackage{clipboard}

\usepackage{thmtools}
\usepackage{thm-restate}

\newtheorem{lemma}{Lemma}

\newtheorem{corollary}{Corollary}[lemma]

\usepackage{xcolor}
\definecolor{DarkGreen}{RGB}{1,50,32}

\usepackage{silence}
\usepackage{chngcntr}
\counterwithin*{equation}{section}
\counterwithin*{equation}{subsection}

\begin{document}
\ActivateWarningFilters[pdftoc]
\newcommand{\defeq}{:=}
\newcommand{\eps}{\varepsilon}
\newcommand{\ReturnCode}{\textbf{return}}

\newcommand{\sidford}[1]{{\color{red} \textbf{Aaron}: #1}} 
\newcommand{\liam}[1]{{\color{red} \textbf{Liam}: #1}} 
\newcommand{\avi}[1]{{\color{purple} \textbf{Avi}: #1}} 
\newcommand{\vvw}[1]{{\color{Plum} \textbf{Virgi}: #1}} 
\newcommand{\uri}[1]{{\color{red} \textbf{Uri}: #1}} 
\newcommand{\new}[1]{{\color{blue} #1}} 

\newcommand{\blue}[1]{{\color{blue}#1}}

\newcommand{\codestyle}[1]{\texttt{#1}}
\newcommand{\Initialize}{\mbox{\codestyle{Initialize}}}
\newcommand{\Cycle}{\mbox{\codestyle{Cycle}}}
\newcommand{\RT}{\mbox{\codestyle{RT}}}
\renewcommand{\L}{\mbox{\codestyle{L}}}
\newcommand{\CycleOdd}{\codestyle{CycleOdd}}
\newcommand{\BallOrCycle}{\codestyle{BallOrCycle}}
\newcommand{\ClusterOrCycleBounded}{\codestyle{ClusterOrCycleBounded}}
\newcommand{\ClusterOrCycle}{\codestyle{ClusterOrCycle}}
\newcommand{\SimpleCycle}{\codestyle{SimpleCycle}}
\newcommand{\Next}{\codestyle{Next}}
\newcommand{\Sample}{\codestyle{Sample}}
\newcommand{\Dijkstra}{\codestyle{Dijkstra}}
\newcommand{\Preprocess}{\codestyle{Preprocess}}
\newcommand{\HashTable}{\codestyle{HashTable}}
\newcommand{\Heap}{\codestyle{Heap}}
\newcommand{\RelaxNext}{\codestyle{RelaxNext}}

\newcommand{\PreprocessGraph}{\codestyle{Initialize}}
\newcommand{\Route}{\codestyle{Route}}
\newcommand{\TreeRoute}{\codestyle{TreeRoute}}
\newcommand{\N}{\mathbb{N}}
\newcommand{\MinCycle}{\codestyle{MinCycle}}
\newcommand{\Query}{\codestyle{Query}}
\newcommand{\Ball}{\codestyle{Ball}}
\newcommand{\DistanceOracle}{\codestyle{TZ-DistanceOracle}}
\newcommand{\SparseOrCycle}{\codestyle{SparseOrCycle}}
\newcommand{\Intersection}{\codestyle{Intersection}}
\newcommand{\CycleAdditive}{\codestyle{CycleAdditive}}
\newcommand{\GenerateSi}{\codestyle{ComputeS}}

\newcommand{\codeNull}{\codestyle{null}}
\newcommand{\codeYes}{\codestyle{Yes}}
\newcommand{\codeNo}{\codestyle{No}}
\newcommand{\codeAnd}{~ \mathrm{and} ~}
\newcommand{\codeOr}{~ \mathrm{or} ~}
\newcommand{\wt}{\ell}
\newcommand{\polylog}{{\rm polylog}}
\newcommand{\Cl}{CL}
\newcommand{\CL}{CL}
\newcommand{\cl}{c\ell}

\newcommand{\LE}{\;\le\;}
\newcommand{\EQ}{\;=\;}
\newcommand{\GE}{\;\ge\;}
\newcommand{\Ot}{\tilde{O}}
\newcommand{\stactri}{\stackrel\triangle}

\newcommand{\EE}{\mathbb{E}}
\newcommand{\RR}{\mathbb{R}}

\newcommand{\ExtractMin}{\mathtt{ExtractMin}}
\newcommand{\Insert}{\mathtt{Insert}}
\newcommand{\Null}{\mathtt{null}}

\DeclarePairedDelimiter{\ceil}{\lceil}{\rceil}
\DeclarePairedDelimiter{\floor}{\lfloor}{\rfloor}
\DeclarePairedDelimiter{\pair}{\langle}{\rangle}

\begin{titlepage}
\pagestyle{empty}

\date{}
\title{Improved girth approximation in  weighted undirected  graphs}

\author{Avi Kadria\thanks{Department of Computer Science, Bar Ilan University, Ramat Gan 5290002, Israel. E-mail {\tt avi.kadria3@gmail.com}.} \and Liam Roditty\thanks{Department of Computer Science, Bar Ilan University, Ramat Gan 5290002, Israel. E-mail {\tt liam.roditty@biu.ac.il}. Supported in part by BSF grants 2016365 and 2020356.} \and Aaron Sidford\thanks{Departments of Management Science and Engineering and Computer Science, Stanford University, Stanford, CA, 94305, USA. E-mail {\tt sidford@stanford.edu}. Supported in part by BSF grant no.\ 2016365, a Microsoft Research Faculty Fellowship, NSF CAREER Award CCF-1844855, NSF Grant CCF-1955039, a PayPal research award, and a Sloan Research Fellowship
}  \and
Virginia Vassilevska Williams\thanks{Department of Electrical Engineering and Computer Science and CSAIL, MIT, Cambridge, MA, USA. E-mail {\tt virg@mit.edu}. 
Supported in part by NSF CAREER Award 1651838, NSF Grants  CCF-1909429 and CCF- 2129139, BSF grants 2016365 and 2020356, a Google Research Fellowship and a Sloan Research Fellowship.} \and Uri Zwick\thanks{Blavatnik School of Computer Science, Tel Aviv University, Tel Aviv 6997801, Israel. E-mail {\tt zwick@tau.ac.il}. Supported in part by BSF grants 2016365 and 2020356.}}
\maketitle

\begin{abstract}
Let $G = (V,E,\ell)$  be a $n$-node $m$-edge weighted undirected graph, where $\ell: E \rightarrow (0,\infty)$ is a real \emph{length} function defined on its edges, and let $g$ denote the girth of $G$, i.e., the length of its shortest cycle.  We present an algorithm that, for any input, integer $k \geq 1$, in $O(kn^{1+1/k}\log{n}  + m(k+\log{n}))$ expected time finds a cycle of length at most $\frac{4k}{3}g$. This algorithm nearly matches a $O(n^{1+1/k}\log{n})$-time algorithm of \cite{KadriaRSWZ22}
which applied to unweighted graphs of girth $3$. For weighted graphs, this result also improves upon the previous state-of-the-art algorithm that
in $O((n^{1+1/k}\log n+m)\log (nM))$ time, where $\ell: E \rightarrow [1, M]$ is an integral length function, finds a cycle of length at most $2kg$~\cite{KadriaRSWZ22}. 
For $k=1$ this result improves upon the result of Roditty and Tov~\cite{RodittyT13}.
\end{abstract}

\tableofcontents

\thispagestyle{empty}
\clearpage
\end{titlepage}

\pagenumbering{arabic}
\setcounter{page}{1}

\section{Introduction}
\label{sec:intro} 

The length of a shortest cycle, known as the \emph{girth} of the graph, is a key parameter often used to shed light on the structure of graph theory problems (e.g., \cite{lazebnik1997structure,osthus2001almost,hoppen2016properties}). 
Correspondingly, the problem of computing a shortest cycle in an undirected graph is fundamental in algorithmic graph theory and has been studied extensively for decades (e.g., \cite{ItaiR78,AYZ97,YusterZ97,LingasL09,DBLP:journals/siamdm/Ducoffe21} and more).
Prominent special cases, e.g., detecting triangles in graphs, are foundational to algorithm design and complexity theory and are useful in practice as well (e.g., \cite[Chapter~3]{kleinbergbook}).

Given the prominence of shortest cycle and girth computation problems, extensive effort has gone into developing fast algorithms for them. The best known runtime for computing the girth of an $n$-vertex, $m$-edge unweighted graph is $O(\min\{n^\omega,mn\})$f \cite{ItaiR78}, where $\omega<2.373$ is the matrix multiplication exponent. For graphs with nonnegative integer edge lengths bounded by $M$, the best known runtime is $\tilde{O}(\min\{Mn^\omega,mn\})$~\footnote{$\tilde{O}(f(n))$ denotes $O(f(n)\textrm{poly}\log(n))$.} \cite{RodittyW11}. 
For arbitrary edge lengths and no negative cycles, the fastest algorithms solve All-PPairs Shortest Paths (APSP), whose fastest running time to date is $O(\min\{mn+n^2\log\log n,n^3/\exp(\sqrt{\log n})\})$ \cite{Pettie04,ryanapsp}.
Improving upon these bounds significantly would constitute a major breakthrough in algorithm design and fine grained complexity; \cite{WilliamsW18,lincolnsoda} clarified this hardness and proved that such improvements would contradict the APSP hardness hypothesis\footnote{The APSP hypothesis states that no $O(n^{3-\eps})$ algorithm exists for computing APSP in a general weighted graph}.

Given the above hardness for exact girth computation, it is natural to ask which trade-offs between running times and approximation ratios are possible. For unweighted graphs, where all edge lengths are~$1$, there has been extensive work on this question, and there are a variety of results depending on the graph's girth and the desired approximation \cite{LingasL09, RodittyT13, RodittyW12, DahlgaardKS17a, KadriaRSWZ22}. Perhaps most relevant to this paper, Kadria et al.~\cite{KadriaRSWZ22}
presented a collection of new algorithms for girth approximation, including an algorithm that for any input unweighted $n$-vertex undirected graph of girth $g$ and an integer parameter $k\ge 1$ finds a cycle of length at most $2k\cdot \lceil g/2\rceil$ in $O(n^{1+1/k}\log n)$ time. This result improved upon a result of Dahlgaard, Knudsen, and St\"{o}ckel~\cite{DahlgaardKS17a} that provided an algorithm which in the same running time of $O(n^{1+1/k} \log n)$ finds a cycle of length at most $2^k g$, for $k \geq 2$, w.h.p. 

Less is known for the weighted problem of approximating the girth of $n$-vertex $m$-edge undirected graphs with edge lengths in the range $[1, M]$ and girth $g$. Until recently, the state-of-the-art running time for this problem included a result of Roditty and Tov~\cite{RodittyT13} which obtained an $\Ot(n^2 \log M)$-time $4/3$-approximation algorithm which improved upon a $2$-approximation algorithm of Lingas and Lundell~\cite{LingasL09}, as well as a result of Ducoffe \cite{DBLP:journals/siamdm/Ducoffe21} which obtained an $O(m+n^{5/3}\polylog M)$-time $2$-approximation algorithm. Unlike in the case of unweighted graphs, no general trade-off between running time and approximation quality was known. 

Kadria et al.~\cite{KadriaRSWZ22} obtained the first running time versus approximation tradeoff for weighted girth computation. They presented an algorithm that for any integer $k \geq 1$ finds a cycle of length at most $2k\cdot g$ in $O((n^{1+1/k}\log n+m)\log (nM))$ time. For $k=1$, this result offers a worse trade-off than both Roditty and Tov~\cite{RodittyT13} and Ducoffe \cite{DBLP:journals/siamdm/Ducoffe21}. Additionally, the approximation quality achieved by this result is almost twice as large as that achievable for unweighted graphs with a comparable running time.

In light of these results, the central question we ask in this paper is: 
\begin{center}
{\em Is it possible to design a single algorithm that yields improved runtime\\ versus approximation quality trade-offs for undirected weighted graphs? }\end{center}
Our main result is the following theorem:

\begin{restatable}[Improved girth approximation]{theorem}{thmalg}
\label{thm:approx}
Let $G = (V,E,\ell)$ be a weighted, undirected graph,  where $\ell: E \rightarrow (0,\infty)$. Let $g$ be the unknown girth of $G$. For every integer $k\geq 1$\footnote{Throughout the paper, we assume that $k\le \log{n}$ since further increasing $k$ no longer reduces the running time}, there is an algorithm whose expected running time is $O(kn^{1+1/k}\log{n}  + m\log{n})$
that finds a cycle $C$ such that $\ell(C)\leq \frac{4k}{3}g$.
\end{restatable}

Among the tools we use to prove the theorem include: approximate distance oracles \cite{ThorupZ05}, Spira's single-source shortest paths algorithm \cite{Spira73}, a problem related to 2-dimensional orthogonal range reporting from computational geometry, and an extension of ideas from Kadria et al.~\cite{KadriaRSWZ22}.

Our result (up to logarithmic factors in running time) strictly improves upon Kadria et al.~\cite{KadriaRSWZ22}, computing a $(\frac{4k}{3}g)$-approximation (rather than $2kg$) in a comparable running time. 
Furthermore, the approximation quality versus runtime trade-off matches that of Roditty and Tov~\cite{RodittyT13} for $k = 1$.

We note that improving beyond the approximation factor of $4/3$ approximation of Roditty and Tov~\cite{RodittyT13} in the same running time faces a barrier: any $(4/3-\eps)$-approximation algorithm for $\eps>0$ would be able to detect whether a graph contains a triangle, and via~\cite{WilliamsW18}, we know that triangle detection is closely related to BMM. 
Thus, obtaining a quadratic time combinatorial $(4/3-\eps)$-approximation algorithm would contradict the BMM Hypothesis~\cite{WilliamsW18}.
Moreover, our tradeoff matches the known tradeoff of~\cite{KadriaRSWZ22} for unweighted graphs with $g=3$.
This suggests that our scheme might be of the right form.

Beyond improving upon \cite{KadriaRSWZ22} in approximation quality and matching that of Roditty and Tov~\cite{RodittyT13}, we remark that in contrast to these prior results, the runtime of Theorem~\ref{thm:approx} is strongly polynomial (the running time does not depend on $M$).\footnote{We remark that Roditty and Tov~\cite{RodittyT13} also presented an algorithm with $4/3+\epsilon$ approximation and $\tilde{O}((1/\epsilon) n^2)$ running time, and Ducoffe \cite{DBLP:journals/siamdm/Ducoffe21} also presented an algorithm with $(2+\eps)$ approximation and $\Ot(\polylog(1/\epsilon) n^{5/3} + m)$.}  Consequently, Theorem~\ref{thm:approx} can be applied to graphs with arbitrary real positive edge lengths (assuming that addition and comparisons still take constant time), rather than bounded integer values. We thus obtain a strict improvement over the result of Roditty and Tov~\cite{RodittyT13}: a trade-off curve that works for graphs with real edge weights.

Interestingly, to obtain this result, we depart from the approach of Kadria et al.~\cite{KadriaRSWZ22}. Although the running time of the $2k$-approximation algorithm of \cite{KadriaRSWZ22} was $\tilde{O}(m + n^{1 + 1/k})$, if the edges incident to each vertex were given in non-decreasing order of their lengths, the algorithm could run in $\tilde{O}(n^{1 + 1/k})$ time, which is \emph{sublinear} in the input size for sufficiently dense graphs (and may not need to examine all the edges). To obtain this runtime, the algorithm of \cite{KadriaRSWZ22} is limited in how it accesses the graph. It only accesses the edges incident on a given vertex~$u$ in sequential order. That is, it only accesses the $i$-th edge of~$u$, in non-decreasing order of length, after accessing the previous $i-1$ edges of~$u$. 

We show that this property of \cite{KadriaRSWZ22}'s algorithm necessitated a weaker approximation ratio than what we obtain. We show that if the Erd\"{o}s girth conjecture\footnote{The Erd\"{o}s girth conjecture states that there exists graphs with $\Omega(n^{1+1/k})$ edges with girth $\ge 2k+2$, (see e.g., \cite{Erdos64})} holds, then any algorithm that accesses the edges of a weighted graph in the above sequential fashion and makes only $o(n^{1+1/k})$ queries can at best return a $(2k+2)$-approximation. That is, the $2k$-approximation algorithm of Kadria et al.~\cite{KadriaRSWZ22} is optimal given the (Erd\"{o}s) girth conjecture and how it accesses the graph.

\begin{restatable}[Query lower bound]{theorem}{qlower}\label{thm:lower_bound}
Assume that the girth conjecture holds for an integer $k \geq 1$. 
Then for that $k$ and any real value $\tau>0$, every deterministic
algorithm that, when run on an $n$-vertex weighted undirected graph, accessed using the edge oracle model outlined above, computes a cycle $C$ with $\ell(C) \leq (2k+2-\tau)g$,
must make at least $\Omega(n^{1+1/k})$ queries on some graphs.
\end{restatable}

\section{Preliminaries}
\label{sec:prelim}

\subsection{Basic concepts}
Let $G = (V, E,\ell)$  be a weighted undirected graph, where $\ell: E \to (0,\infty)$ is a real \emph{length} function defined on its edges. Let $n=|V|$ and $m=|E|$. 
The graph is represented using an adjacency list representation. We assume that the edges incident on a vertex~$u$ are sorted in a non-decreasing order of length. (If not, this can be easily computed in $O(m\log n)$ time.) All graphs considered are assumed to be connected.

For all $u,v\in V$, we let $\delta_G(u,v)$ be the \emph{distance} from~$u$ to~$v$ in~$G$, i.e., the smallest length of a path from~$u$ to~$v$ in $G$. The length $\ell(P)$ of a path $P$ is the sum of the lengths of its edges, i.e., $\ell(P)=\sum_{e\in P}\ell(e)$. (We usually consider a path~$P$ to be a set of edges, but occasionally we also think of it as a set of vertices.) A path from~$u$ to~$v$ is a \emph{shortest path} if and only if $\ell(P)=\delta_G(u,v)$. Since $G$ is undirected, $\delta_G(u,v)=\delta_G(v,u)$ for every $u,v\in V$. When the graph~$G$ is clear from the context, which will almost always be the case, we write $\delta(u,v)$ instead of $\delta_G(u,v)$.

A tree $T$ rooted at $u$ and containing the vertices of a set $U$ is said to be a \emph{shortest paths tree} from~$u$ to the vertices of~$U$ if, for every $v\in U$, the path from~$u$ to~$v$ in~$T$ is a shortest path from~$u$ to~$v$ in~$G$.

If $u\in V$ and $A\subseteq V$, we let $\delta(A,u)=\delta(u,A)=\min_{v\in A} \delta(u,v)$ denote the  \emph{distance from~$u$ to the set~$A$}. (If $A=\emptyset$, then $\delta(u,A)=+\infty$.)
We define the distance $\delta(u,(v,w))$ from a vertex~$u\in V$ to an edge $(v,w)\in E$ as follows: $\delta(u,(v,w))=\min\{\delta(u,v),\delta(u,w)\}+\ell(v,w)$. (Note that $\delta(u,(v,w))=\delta(u,\{v,w\})+\ell(v,w)$. Here $\{v,w\}$ is a set of two vertices.) 

The \emph{girth} $g$ of a graph $G=(V, E,\ell)$ is the length of a shortest simple cycle in~$G$. (If $(u,v)$ is an edge, then $(u,v,u)$ is not considered to be a cycle.) The length of a cycle~$C$ is the sum of the lengths of the edges on~$C$, i.e., $\ell(C)=\sum_{e\in C}\ell(e)$. We also let $M(C)=\max_{e\in C} \ell(e)$ be the maximum edge length on~$C$.

\subsection{Balls}

Given a graph $G=(V,E,\ell)$, a vertex $u\in V$ and $r>0$, we define the \emph{ball graph} $G_r(u) = (V_r(u),E_r(u))$ of \emph{radius} $r$ around~$u$ as follows:
\begin{align*}
    V_r(u) \defeq \{ v\in V \mid \delta(u,v)\leq r\} 
    \text{ and }
    E_r(u) \defeq \{ e\in E \mid \delta(u,e)\leq r\} \;.
\end{align*}
Note that $G_r(u)$ is not necessarily the same as $G[V_r(u)]$, the subgraph of $G$ induced by the vertex set $V_r(u)$ of $G_r(u)$. 
For example, $G[V_r(u)]$  may include edges with length greater than $r$, whereas such edges are excluded from $E_r(u)$.

We let $G_{<r}(u)=(V_{<r}(u),E_{<r}(u))$ denote the \emph{open ball graph of radius~$r$ around~$u$}. The definitions of $V_{<r}(u)$ and $E_{<r}(u)$ are identical to those of $V_{r}(u)$ and $E_{r}(u)$ with the weak inequalities $\delta(u,v)\leq r$ and $\delta(u,e)\leq r$ replaced by strict inequalities.

The following general and simple lemma is useful in proving the correctness of our algorithms and is a natural extension of the main lemma of \cite{KadriaRSWZ22}. 

\begin{lemma} \label{lem:main}
Let  $G = (V, E,\ell)$ be a weighted undirected graph, $C$ a cycle in~$G$, $u\in V$, and $r>0$. If $V_r(u)\cap C\ne \emptyset$, then $C\subseteq G_{r+\frac{1}{2}(\ell(C)+M(C))}(u)$.
\end{lemma}

\begin{proof} 
Let $v\in V_r(u)\cap C$. By definition $\delta(u,v)\le r$. 
Let $(x,y)\in C$. Assume, without loss of generality, that $\delta(v,x)\le \delta(v,y)$. As $\delta(v,x)+\ell(x,y)+\delta(v,y)\le \ell(C)$, we get that $\delta(v,x)\le \frac12(\ell(C)-\ell(x,y))$. Thus
\begin{align*}
    \delta(u,(x,y)) &\LE \delta(u,v)+\delta(v,x)+\ell(x,y) \\
    &\LE r + \frac12(\ell(C)-\ell(x,y)) + \ell(x,y) \\
    &\EQ r+\frac12(\ell(C)+\ell(x,y)) 
     \LE r+\frac12(\ell(C)+M(C))\;.
\end{align*}
Thus, $(x,y)\in E_{r+\frac{\ell(C)+M(C)}{2}}(u)$, for every $(x,y)\in C$ and therefore $C\subseteq G_{r+\frac{\ell(C)+M(C)}{2}}(u)$, as required.
\end{proof}

\subsection{Clusters}\label{sub-clusters}

Let $G=(V,E,\ell)$ be a weighted undirected graph and let $V=A_0\supseteq A_1 \supseteq A_2 \supseteq \dots \supseteq A_k = \emptyset$ be a hierarchy of vertex sets, where $k \ge 1$. If $u\in A_i\setminus A_{i+1}$, then following \cite{ThorupZ05} we define the \emph{cluster} of~$u$ in~$G$ to be the graph $\Cl(u)=(\Cl_V(u),\Cl_E(u))$, where
\begin{align*}
    \Cl_V(u) &\EQ \{ v\in V \mid \delta(u,v)<\delta(v, A_{i+1}) \} \;, \\
    \Cl_E(u) &\EQ \{ (v,w)\in E \mid \delta(u,v) + \ell(v,w) < \delta(w, A_{i+1}) \} \;.
\end{align*}
Note that, unlike \cite{ThorupZ05}, we define the cluster $\Cl(u)$ to be a graph and not just a vertex set.

For a vertex $u\in V$, we let $a(u)=i$ where $u\in A_i\setminus A_{i+1}$. For any $u\in V$ and $0\le i<k$, we let $p_i(u)=\arg \min_{v\in A_i} \delta(u,v)$, i.e., $p_i(u)$ is a vertex of~$A_i$ closest to~$u$ (ties are broken lexicographically). 

\begin{lemma}\label{L-shortest-path-cluster}
Let $u\in A_i\setminus A_{i+1}$. If $v\in \Cl_V(u)$ and $P$ is a shortest path from~$u$ to~$v$, then all the vertices and edges on~$P$ are also in $\Cl(u)$.
\end{lemma}

\begin{proof}
Let $x$ be a vertex on the shortest path~$P$ from~$u$ to~$v$. Assume, for contradiction, that $x\notin\Cl_V(u)$. Let $w=p_{i+1}(x)\in A_{i+1}$. Then, $\delta(w,x)=\delta(x,A_{i+1})\le\delta(u,x)$. It follows that $\delta(w,v)\le \delta(w,x)+\delta(x,v)\le\delta(u,x)+\delta(x,v)=\delta(u,v)$, contradicting that $v\in \Cl_V(u)$. The proof for the edges on~$P$ is similar.
\end{proof}

Clusters have especially nice properties when the hierarchy $V=A_0\supseteq A_1 \supseteq A_2 \supseteq \dots \supseteq A_k = \emptyset$ is obtained using random sampling. Lemma~\ref{L-cluster-size} gives one such property and is proven in \cite{ThorupZ05} using a simple probabilistic argument. 

\begin{lemma}[\cite{ThorupZ05}]\label{L-cluster-size}
If $A_{i+1}$, for $i=0,1,\ldots,k-2$, is obtained by including each vertex of $A_i$ independently with probability $n^{-1/k}$, then $\EE[\sum_{u\in V}|\Cl_V(u)|]=O(kn^{1+1/k})$.
\end{lemma}

Thorup and Zwick \cite{ThorupZ05} describe a simple modification of Dijkstra's algorithm that allows $\Cl(u)$ to be constructed in $\tilde{O}(|E(\Cl_V(u))|)$ time, where $E(\Cl_V(u))$ is the set of all edges in~$G$ incident on a vertex of $\Cl_V(u)$.\footnote{Using the ideas of the next section we can actually improve the running time needed to compute $\Cl(u)$ to $\tilde{O}(|\Cl_E(u)|$, but this may still be too slow.} All clusters can therefore be constructed in $\tilde{O}(kmn^{1/k})$ time. This is too slow for us as we are aiming for a running time of $\tilde{O}(kn^{1+1/k})$. Our girth approximation algorithm constructs most clusters only partially, until a cycle in them is detected. This is described in the next section.

\subsection{Initialization}\label{sub-initialization}

Next, we describe an initialization algorithm used by our girth approximation algorithm described in Section~\ref{sec:weighted_approx}. The initialization algorithm, $\Initialize(G,k)$, see Algorithm~\ref{A-Initialize}, receives the input graph $G=(V,E,\ell)$ and the parameter~$k\ge 1$. 
The algorithm starts by sampling the vertex hierarchy $V=A_0\supseteq A_1 \supseteq A_2 \supseteq \dots \supseteq A_k = \emptyset$. 

Next, it initializes two empty \emph{hash tables} $d$ and $\pi$ used to store the distances and shortest paths already computed by the algorithm. When the algorithm discovers a distance $\delta(u,v)$ between two vertices $u,v\in V$, it inserts the pair $(u,v)$ into the hash table~$d$ with value 
$\delta(u,v)$. For brevity, we write this as $d(u,v)\gets \delta(u,v)$. When we want to check whether $\delta(u,v)$ was already computed, we search $(u,v)$ in the hash table~$d$. If $(u,v)$ is found we retrieve $\delta(u,v)$. For brevity, we interpret $d(u,v)$ as a search for $(u,v)$ in the hash table~$d$. The search returns $\delta(u,v)$ if $(u,v)$ is in the table, or $+\infty$, if $(u,v)$ is not in the table, i.e., $\delta(u,v)$ is not yet known to the algorithm. (We assume that $d(u,v)$ searches both $(u,v)$ and $(v,u)$, or more efficiently, that all pairs $(u,v)$ stored in the table satisfy $u<v$.)

The hash table $\pi$ is similarly used to represent the shortest paths already found by the algorithm. If $d(u,v)<\infty$, then $\pi(u,v)$ is the last edge on a shortest path from~$u$ to~$v$. Thus, if $d(u,v)<\infty$ and $\pi(u,v)=(w,v)$ then $d(u,v)=d(u,w)+\ell(w,v)$.

We assume that each operation on the hash tables $d$ and $\pi$ takes constant expected time, as this can be achieved using standard hashing techniques. 

\begin{algorithm2e}[t] 
    \DontPrintSemicolon
 	\caption{$\Initialize(G = (V,E,\ell), k)$} \label{A-Initialize}\label{A-Clusters}
 	\BlankLine
    $A_0\gets V$ ; $A_k \gets \emptyset$ \;
    \BlankLine
    \For {$i\gets 1$ {\bf to} $k-1$} { \label{L-initialize-Ai}
    	$A_i\gets \Sample(A_{i-1}$,$n^{-1/k}$) \;
    }
    \BlankLine
    $d\gets\HashTable()$\qquad \tcp{\rm Used to store computed
    distances.} 
    $\pi\gets\HashTable()$\qquad \tcp{\rm Used to store computed
    shortest paths.} 
    \BlankLine
    \For{$i\gets 1$ {\bf to} $k-1$}
    {
        $\Dijkstra(G,A_i)\quad$ \tcp{\rm Finds $\delta(u,A_i)$ and $p_i(u)$ for every $u\in V$.} 
        \For{$u\in V$} {
        $d(p_i(u),u)\gets \delta(u,A_i)$ \;
        }
    }
    \BlankLine
    $\Preprocess(G)$ [Section~\ref{sub-next}] \;

\end{algorithm2e}

The algorithm then computes the distances $\delta(u,A_i)$, for every $u\in V$ and $0\le i<k$. This is easily done by adding an auxiliary vertex $s_i$, connecting it with $0$-length edges to all vertices of~$A_i$ and then running Dijkstra from~$s_i$, as done in \cite{ThorupZ05}. This also computes $p_i(u)=\arg\min_{v\in A_i} \delta(u,v)$ for every $u\in V$ and $0\le i<k$ and a corresponding shortest path from~$u$ to~$p_i(u)$. 

Finally, \Initialize\ calls \Preprocess\ that performs preprocessing operations on the adjacency lists of all vertices. This processing includes sorting each adjacency list in non-decreasing order of edge length, and for every $0\le i<k$ building a binary tree on the edges of the vertex, as explained in Section~\ref{sub-next}. The total cost of all these preprocessing operations is $O(m\log n)$. 

\begin{lemma}\label{L-Init}
$\Initialize(G,k)$ takes $O((m+kn)\log n)$ time.
\end{lemma}

\begin{proof}
The $k$ calls to Dijkstra's algorithm take $O(k(m+n\log n))$ time. Preprocessing the adjacency lists takes $O(m\log n)$. 
\end{proof}
Throughout the paper, a graph $G$ is said to be \emph{initialized} if the procedure $\Initialize(G, k)$ has already been called on it.
\subsection{Cycle detection and compact cycle representation}\label{sub-cycle}

When the girth approximation algorithm discovers an edge $(v,w)$ such that both distances $\delta(u,v)$ and $\delta(u,w)$ are known, for some $u\in V$, it checks whether $(v,w) \notin \{\pi(u,v),\pi(u,w)\}$. If so, a cycle is detected. (Recall that $\pi(u,v)$ and $\pi(u,w)$ are last edges on shortest paths from~$u$ to~$v$ and~$w$, respectively.) The actual cycle is composed of the shortest paths from~$u'$ to~$v$ and~$w$, where $u'$ is the LCA (lowest common ancestor) of~$v$ and~$w$ in the shortest paths tree rooted at~$u$, and the edge $(v,w)$.

This cycle and its length can be easily found in time proportional to the number of edges on the cycle. In some cases, faster running times are desired. We thus succinctly represent the discovered cycle by the triplet $(u,v,w)$ and use $\delta(u,v)+\delta(u,w)+\ell(v,w)=d(u,v)+d(u,w)+\ell(v,w)$ as an upper bound on its length. \footnote{In principle, we can use an LCA data structure to find $u'$ and the actual length of the cycle in $O(1)$ time. This complicates the algorithm and does not lead to improved results.}

\section{Algorithm \texorpdfstring{$\ClusterOrCycle$}{ClusterOrCycle}} \label{sec:tools}

In this section, we describe an algorithm $\ClusterOrCycle$ (Algorithm~\ref{A-ClusterOrCycle}) that assumes that the graph has been initialized, and receives as input a vertex $u\in U$. The algorithm either returns the cluster $\Cl(u)$ or a cycle, as stated in the following lemma:

\begin{lemma}\label{L-ClusterOrCycle} Let $G=(V,E,\ell)$ be an initialized weighted undirected graph and let $u\in V$. If $\Cl(u)$ is a tree, then $\ClusterOrCycle(u)$ finds $\Cl(u)$, the distance $\delta(u,v)$ for each $v\in \Cl(u)$, and a tree of shortest paths from~$u$ to all vertices of $\Cl(u)$. Otherwise, if $r>0$ is the smallest number such that $\Cl(u)\cap G_r(u)$ contains a cycle, then $\ClusterOrCycle(u)$ returns a description of a cycle in $\Cl(u)\cap G_r(u)$ whose length is at most~$2r$. Furthermore, it returns $\Cl(u)\cap G_{<r}(u)$ and a tree containing shortest paths from $u$ to all vertices of $\Cl(u)\cap G_{<r}(u)$. $\ClusterOrCycle(u)$ can be implemented in $O(|\Cl_V(u)|\log n)$ time.
\end{lemma}

$\ClusterOrCycle(u)$ goes through the appropriate steps to construct the cluster $\Cl(u)$, but stops early whenever a cycle in $\Cl(u)$ is encountered. This ensures that $\ClusterOrCycle(u)$ can be implemented in time proportional to the number of vertices in $\Cl(u)$, and not to the number of edges in $\Cl(u)$, which would have been too expensive. It uses a modification of Spira's \cite{Spira73} single-source shortest paths algorithm. 

Spira's algorithm assumes that the edges incident on each vertex are sorted in non-decreasing order of length. It may be viewed as a lazy version of Dijkstra's \cite{Di59} algorithm. In certain cases, it may find distances to all vertices without examining all edges. (This is possible as the adjacency lists of all vertices are assumed to be sorted by length. A recent application of Spira's algorithm can be found in \cite{WiZw15}.) 

When Dijkstra's algorithm discovers the distance from the source~$u$ to a new vertex~$v$, it immediately relaxes all the outgoing edges $(v,w)$ of~$v$. Spira's algorithm only relaxes the first outgoing edge of~$v$. The heap~$Q$ used by Spira's algorithm contains edges rather than vertices. 
Relaxing an edge $(v,w)$ amounts to inserting it into~$Q$ with key $d(u,v)+\ell(v,w)$. 
The algorithm also maintains a set~$U$ of vertices whose distance from the source~$u$ has already been found. Initially $U=\{u\}$. In each iteration, Spira's algorithm extracts an edge $(v,w)$ of minimum key from the heap~$Q$. If $w\notin U$, it adds $w$ to~$U$ and sets $d(u,w)\gets d(u,v)+\ell(v,w)$ which is guaranteed to be the distance from $u$ to~$w$. It now relaxes the first edge of~$w$ and the next edge of~$v$, i.e., the edge following $(v,w)$ in the sorted adjacency list of~$v$, if there is such an edge. If $w\in U$, the algorithm simply relaxes the next edge of~$v$. When $U=V$, the algorithm stops, even if there are still edges left in the heap~$Q$ and even if some edges were not examined yet.

The correctness of Spira's algorithm follows easily from the correctness of Dijkstra's algorithm, or can be proved directly using the same ideas used to prove the correctness of Dijkstra's algorithm. 

Algorithm $\ClusterOrCycle(u)$, shown as Algorithm~\ref{A-ClusterOrCycle}, uses the following modification of Spira's algorithm. It starts constructing $\Cl(u)$. The set $\Cl(u)$ denotes the set of vertices of the cluster discovered so far. Initially $\Cl(u)=\{u\}$. When the first edge $(v,w)$ for which $w\in \Cl(u)$ is extracted from the heap~$Q$, the algorithm stops as a cycle in $\Cl(u)$ is discovered, and the algorithm returns the discovered cycle. 

A non-trivial complication arises from the fact that we want $\ClusterOrCycle(u)$ to only examine edges that belong to $\Cl(u)$. Furthermore, for a correct implementation of Spira's algorithm, we need to examine these edges in non-decreasing order of length.

For the high-level description of Algorithm $\ClusterOrCycle(u)$, we assume that we have a function $\Next(u,v)$ that given a vertex~$v$ already known to be in $\Cl(u)$ gives us the next incident edge~$(v,w)$ of~$v$ that leads to a vertex~$w$ also in~$\Cl(u)$, in non-decreasing order of length. If there is no such next edge, then $\Next(u,v)$ returns $\Null$. The implementation of $\Next(u,v)$ is described in Section~\ref{sub-next}. It is shown there that it can be implemented in $O(\log n)$ time.

$\ClusterOrCycle(u)$ uses $\Next(u,v)$ via a function $\RelaxNext(u,v)$, see Algorithm~\ref{A-RelaxNext}, that uses $\Next(u,v)$ to extract the next eligible edge~$e$, if there is any, and relax it, i.e., add $e$ to the heap~$Q$ with key $d(v)+\ell(e)$. 

\begin{algorithm2e}[t]
    \DontPrintSemicolon
    \caption{$\ClusterOrCycle(u)$}\label{A-ClusterOrCycle}
    $d(u,u)\gets 0$ \; $\pi(u,u)\gets \Null$ \; 
    $\cl(u) \gets \{u\}$ \;
    \BlankLine
    $Q\gets \Heap()$ \; 
    $\RelaxNext(u,u)$ \;
    \BlankLine
    \While{$Q\ne\emptyset$}
    {
        \BlankLine
        $(v,w)\gets Q.\ExtractMin()$ \;
        \BlankLine
        \If{$w\in \cl(u)$}
        {\Return $\langle \, (u,v,w)\,,\,d(u,v)+\ell(v,w)+d(u,w)\,\rangle$ \;}
        \BlankLine
        $d(u,w) \gets d(u,v)+\ell(v,w)$ \; $\pi(u,w)\gets (v,w)$ \; $\cl(u)\gets \cl(u)\cup\{w\}$ \;
        \BlankLine
        $\RelaxNext(u,v)$ \; 
        $\RelaxNext(u,w)$ \;
    }
    \BlankLine
    \Return $\cl(u)$ 
\end{algorithm2e}

\begin{algorithm2e}[t]\label{A-RelaxNext}
    \DontPrintSemicolon
    \caption{$\RelaxNext(u,v)$}
    $e\gets \Next(u,v)$ \;
    \If{$e\ne \Null$}
    {
        $Q.\Insert(e,d(u,v)+\ell(e))$ \;
    }
\end{algorithm2e}

We end this section with a proof of Lemma~\ref{L-ClusterOrCycle}.

\begin{proof}[Proof of Lemma~\ref{L-ClusterOrCycle}]
$\ClusterOrCycle(u)$ starts running Spira's algorithm on the implicitly represented cluster graph $\Cl(u)$. The algorithm extracts the edges $(v,w)$ of $\Cl(u)$ from the heap~$Q$ in non-decreasing order of their key $d(u,v)+\ell(v,w)$. When the first edge $(v,w)$ reaching a vertex~$w$ of~$\Cl(u)$ is extracted from~$Q$, then $\delta(u,w)=\delta(u,v)+\ell(v,w)$. The distance $d(u,w)$ is set accordingly, and $w$ is added to~$\Cl(u)$, the set of vertices of the cluster discovered so far. If a second edge $(v',w)$ reaches the same vertex~$w$ is extracted from~$Q$, then a cycle is detected and returned.
If $\Cl(u)$ does not contain a cycle, then from the correctness of Spira's algorithm, the algorithm $\ClusterOrCycle$ returns $\Cl(u)$ as required.

Let $r>0$ be the smallest number, as in the statement of the lemma, such that $\Cl(u)\cap G_r(u)$ contains a cycle. As $\Cl(u)\cap G_{<r}(u)$ does not contain a cycle, $\ClusterOrCycle(u)$ finds distances and shortest paths to all vertices of $\Cl(u)\cap G_{<r}(u)$ before a second edge reaching a vertex is found. The algorithm then starts finding vertices of distance exactly~$r$ from~$u$. As $\Cl(u)\cap G_r(u)$ contains a cycle, at some stage a second edge reaching a vertex in $\Cl(u)\cap G_r(u)$ must be found, and Spira's algorithm is aborted. This edge clearly closes a cycle of length at most $2r$, which is returned by the algorithm, as required.

Spira's algorithm spends $O(\log n)$ time on each edge $(v,w)$ it considers. This includes the $O(\log n)$ time taken by $\Next(u,v)$ [Section~\ref{sub-next}] to return the edge, the $O(\log n)$ (or $O(1)$) time needed to insert the edge to the heap~$Q$, and the $O(\log n)$ time needed for extracting it from the heap.
The size of the heap is always at most the number of vertices in $\Cl(u)$, i.e., the vertices of the cluster discovered so far. 
As long as no cycles are found, the number of edges examined by Spira's algorithm is at most $2|\Cl(u)|-1$: the number of edges extracted from~$Q$ is $|\Cl(u)|-1$ and the number of edges in~$Q$ is at most $|\Cl(u)|$. When a cycle is found, the total number of edges examined is at most $2|\Cl(u)|$. The total running time is therefore $O(|\Cl(u)|\log n)=O(|\Cl(u)|\log n)$, as claimed.
\end{proof}

\subsection{Examining cluster edges in non-decreasing order of length}\label{sub-next}
Recall that if $u\in A_i\setminus A_{i+1}$ then $\Cl(u)=(\Cl_V(u),\Cl_E(u))$, where
\begin{align*}
    \Cl_V(u) &\EQ \{ v\in V \mid \delta(u,v)<\delta(v,A_{i+1}) \} \;, \\
    \Cl_E(u) &\EQ \{ (v,w)\in E \mid \delta(u,v)+\ell(v,w)<\delta(w,A_{i+1}) \} \;.
\end{align*}

Algorithm $\Preprocess$, called by \Initialize, defines $k$ \emph{shifted lengths} as follows, $\ell_i(v,w) = \ell(v,w)-\delta(w,A_{i+1})$ for each edge  $(v,w)\in E$, for every $i\in [0,k-1]$. Now, if $u\in A_i\setminus A_{i+1}$ and $v\in \Cl_V(u)$ then $(v,w)\in\Cl_E(u)$ if and only if $\ell_i(v,w)<-\delta(u,v)$. We want to iterate over the edges of~$v$ that satisfy this condition in increasing order of their original length.

Abstractly, we are faced with the following situation. We have a sequence $e_1,e_2,\ldots,e_n$ of items. (In our concrete situation these are the edges incident on some vertex~$v$ and $n$ is the degree of~$v$, where we already know that $v\in \Cl(u)$ and also have $d(u,v)=\delta(u,v)$.) Each item $e$ has two lengths, $x(e)$ and $y(e)$. (In the concrete case, these are $\ell(e)$ and $\ell_i(e)$, where $u\in A_i\setminus A_{i+1}$.) We are given a bound~$y_0$ and are required to iterate over all items that satisfy $y(e)<y_0$ in non-decreasing order of $x(e)$, until we decide that we do not want to see additional items. (In our case $y(e)=\ell_i(e)$ and $y_0=-\delta(u,v)$.) We want to produce each item in, say, at most $O(\log n)$ time.

This is closely related to the \emph{2-dimensional orthogonal range reporting} problem. In this problem, we are given a collection of $n$ points $(x_j,y_j)$ in the plane. Given four thresholds $a<b$ and $c<d$, we want to return all the points in the box $[a,b]\times[c,d]$, i.e., all the points satisfying $a\le x_j\le b$ and $c\le y_j\le d$. A classical result of Chazelle \cite{Chazelle86}, which improves on a result of Willard \cite{Willard85}, says that this can be done in $O(k+\log n)$ time using $O(n(\log n)/(\log\log n))$ space, where $k$ is the number of points returned.

Our problem is slightly easier, on the one hand, as we have only one threshold~$d$. On the other hand, we want to produce the items satisfying $y_j<d$, one by one, in non-decreasing order of their $x$-coordinate. We are not allowed to first collect all items satisfying $y_j<d$ and then sort them according to their $x$-coordinates, as we may only want to look at the first few points satisfying the condition, or even just the first.


As we have only one threshold, we can solve our problem using ideas borrowed from the \emph{priority search tree} of McCreight \cite{McCreight85}. (These ideas work, in fact, for up to three thresholds.)

We sort the $n$ points according to their $x$-coordinate and put them at the leaves of a binary search tree. (For simplicity, we may assume that~$n$ is a power of~2.) Each node of the tree contains the minimum $y$-coordinate among all the items in its subtree. All these values can be easily computed in $O(n)$ time by letting the value of each vertex be the minimum of the values of its two children.

Given an upper bound $y_0$ we can now easily find the item $(x_j,y_j)$ with the minimum $x$-coordinate that satisfies $y_j<y_0$. First, we check if the minimum $y$-value of the root is less than $y_0$. If not, then there is no point in satisfying the condition. Then, starting at the root, we repeatedly go to the left child if its minimum $y$ value is less than~$y_0$, and to the right child otherwise. The first item can thus be found in $O(\log n)$ time. Similarly, given an item, we can easily find the next item in $O(\log n)$ time. Thus, the first~$k$ items in non-decreasing order of their $x$-coordinates, can be found in $O(k\log n)$ time.


Agarwal \cite{Agarwal22} pointed out a more efficient, but slightly more complicated, solution. Insert the points in non-decreasing order of their $y$-coordinates into a \emph{persistent} red-black tree. (See Sarnak and Tarjan \cite{SarnakT86}.) The keys of the points are their $x$-coordinates. Given a threshold $y_0$, do a binary search on the $y$-coordinates to find the appropriate version of the red-black tree and start listing the items in this tree in non-decreasing order of their $x$-coordinate. Producing the first~$k$ points then takes only $O(\log n + k)$ instead of $O(k\log n)$.

Producing each edge in $O(\log n)$ time is enough for our purposes as we spend $\Omega(\log n)$ time on each edge in any case. Thus, Agarwal's elegant idea does not lead to an improved running time of the whole algorithm.

\section{Girth approximation algorithm}
\label{sec:weighted_approx}

In this section, we prove Theorem~\ref{thm:approx} which we restate for convenience:

\thmalg*

To prove Theorem~\ref{thm:approx} we present an Algorithm $\Cycle$ that receives as an input a weighted undirected graph $G = (V, E,\ell)$ with girth $g$ and an integer parameter $k\ge 1$ and finds a cycle of length at most $\frac{4k}{3}g$. 

The algorithm $\Cycle$ (code in Algorithm~\ref{A-Cycle}) works as follows.  
$\Cycle$ starts by calling $\Initialize(G,k)$. It then sets~$\alpha = 0$ and~$W = \emptyset$. Here, $\alpha$ is an upper bound on the length of the smallest cycle found so far and $W$ is a triplet describing this shortest cycle, as explained in Section~\ref{sub-cycle}.

Next, $\Cycle$ calls $\ClusterOrCycle(u)$ for every $u\in V$. The result of $\ClusterOrCycle(u)$ is the pair 
$\langle \alpha', W'\rangle$. If $\alpha'<\alpha$ then $\alpha$ and $W$ are updated to be  $\alpha'$ and $W'$, respectively. 
Finally, for every $(v,w)\in E$ and every $0 \leq i \leq k-1$, the algorithm checks whether the edge $(v,w)$ closes a cycle in shortest paths tree of $u=p_i(v)$, and if this cycle is shorter than the shortest cycle found so far. More precisely, the algorithm checks whether $d(u,v)$ and $d(u,w)$ are defined by two accesses to the hash table~$d$. If they are defined, they correspond to the actual distances $\delta(u,v)$ and $\delta(u,w)$. Otherwise, they are $+\infty$. Next, the algorithm checks that $\pi(u,v)\ne (w,v)$ and $\pi(u,w)\ne (v,w)$. If this condition holds, then a cycle is indeed formed and $\alpha'=d(u,v)+\ell(v,w)+d(u,w)$ is an upper bound on its length. If $\alpha'<\alpha$ we update $\alpha$ and $W$ accordingly.

\begin{algorithm2e}[t]
    \DontPrintSemicolon
 	\caption{$\Cycle(G = (V,E,\ell), k)$} \label{A-Cycle}
 	\BlankLine
    $\Initialize(G,k)$\;
    \BlankLine
    $\alpha \gets \infty$ ;
    $W \gets \emptyset$\;
    \BlankLine
    \For {$u\in V$}{ \label{L-Alg-Loop1}
          $\langle W',\alpha'\rangle  \gets \ClusterOrCycle(u)$\; \label{L-Alg-ClusterOrCycle}
          \If {$\alpha' < \alpha$} 
          {
          	    $\alpha \gets \alpha'$ ;
          	    $W \gets W'$
          }
       
    }
    \BlankLine
    \For {$(v,w)\in E$}{ \label{L-Alg-Loop2}
    \For{$i\gets 0$ {\bf to} $k-1$}
        {
        $u\gets p_i(v)$ ;
        $\alpha' \gets d(u,v) + \ell(v,w) + d(u,w)$ \;
        \If {$\alpha' < \alpha$ {\bf and} $\pi(u,v) \neq (w,v)$ {\bf and} $\pi(u,w) \neq (v,w)$ }{ 
        \label{L-Cycle-If-pi}
            $\alpha \gets \alpha'$ ;
            $W\gets (u,v,w)$ \;
                  }
        }
    }
    \Return $\langle W,\alpha \rangle$
\end{algorithm2e}

Let $C$ be a shortest cycle in~$G$. We break the correctness proof of $\Cycle$ into two cases: either $M(C)\leq g/3$ or $M(C)> g/3$. (Recall that $M(C)$ is the length of the longest edge on~$C$.).
If $M(C) \leq g/3$, then we show in Lemma~\ref{L-Cycle-Correct-Small} that the first for loop satisfies the desired approximation. 

We begin by considering the case that $M(C) \leq g/3$. We show that if there is  $w \in A_{i}$ that is relatively close to $C$, then either $\Cycle$ finds a cycle within the desired bound or there exists $w' \in A_{i+1}$ that is relatively close to $C$. 

\begin{lemma} \label{L-Cycle-Correct-Small}
Let $C$ be a cycle in $G$ such that $\ell(C)=g$ and $M(C)\leq g/3$. Let $0\leq i\leq  k-1$. If there exists $w \in A_{i}$ such that $\delta(w, C) \le \frac{2i}{3}g$ then either \mbox{\rm $\Cycle$}  finds a cycle of length at most $\frac{4(i+1)}{3}g$, or there exists $w' \in A_{i+1}$ such that $\delta(w', C) \le \frac{2(i+1)}{3}g$.
\end{lemma}
\begin{proof}
Let $0\leq i\leq k-1$ and let $w \in A_{i}$ such that $\delta(w, C) \le \frac{2i}{3}g$. 
If there exists $x \in C$ such that $\delta(x, A_{i+1}) \leq \frac{2(i+1)}{3}g$ then 
there exists $p_{i+1}(x)=w' \in A_{i+1}$ such that $\delta(w', C) \le \frac{2(i+1)}{3}g$, as required. 
Thus, for the rest of the proof we assume that $\delta(x, A_{i+1}) > 2(i+1)g/3$, for every vertex $x\in C$.

We show that $C\subseteq \Cl(w) \cap G_{2(i+1)g/3}(w)$. 
We first show that $C\subseteq  G_{2(i+1)g/3}(w)$.
Let $y=\arg \min_{z\in C}\delta(w, z) $.
Since $\delta(w, y) \le 2ig/3$ we have $y\in V_{2ig/3}(w)\cap C \neq  \emptyset$ and we can apply Lemma~\ref{lem:main} with $r=2ig/3$ and $M(C) \le g/3$ to  get  that $C \subseteq G_{2ig/3 + g/2 + g/6}(w) = G_{2(i+1)g/3}(w)$. By the definition of ball graphs, this implies that $\delta(w, (s,t)) \le 2(i+1)g/3$, for every edge $(s,t)\in C$.

We now show that $C\subseteq \Cl(w)$.
Recall that we are in the case where for every vertex $x\in C$ we have $\delta(x, A_{i+1}) > 2(i+1)g/3$, and thus,  
$\delta(w, (s,t))\le 2(i+1)g/3<d(t, A_{i+1})$ and $\delta(w, (s,t))\leq 2(i+1)g/3< d(s,A_{i+1})$, for every edge $(s,t)\in C$. 

By definition we have $\delta(w, (s,t)) = \min \{\delta(w,s),\delta(w,t)\}+\ell(s,t)$. 
Assume, without loss of generality, that $\delta(w,s) \leq \delta(w,t)$. Thus, $\delta(w,s) + \ell(s,t)=\delta(w,(s,t))<d(t, A_{i+1})$ which implies that $(s,t)\in \Cl_E(w)$ and  $C \subseteq \Cl(w)$.

It follows from Lemma~\ref{L-ClusterOrCycle}  that if $C\subseteq \Cl(w) \cap G_{2(i+1)g/3}(w)$ then  $\ClusterOrCycle(w)$ finds a cycle of length at most $2\cdot 2(i+1)g/3=4(i+1)g/3$, and the claim follows.
\end{proof}

Next, we use Lemma~\ref{L-Cycle-Correct-Small} to prove the following Lemma~\ref{L-Cycle-Correct-Small-2}.

\begin{lemma}\label{L-Cycle-Correct-Small-2}
Let $C$ be a cycle in $G$ such that $\ell(C)=g$ and $M(C)\leq g/3$. Let $0\leq i\leq k-1$. Either \mbox{\rm $\Cycle$} finds a cycle of length at most $\frac{4(i+1)}{3}g$, or there exists $w' \in A_{i+1}$ such that $\delta(w', C) \le \frac{2(i+1)}{3}g$.
\end{lemma}
\begin{proof}
We prove the claim by induction on $i$.
For the base case $i=0$, we have $A_0=V$ and $C\cap A_0\neq \emptyset$. Let $z\in C\cap A_0$. Since $\delta(z,z)=0$, by Lemma~\ref{L-Cycle-Correct-Small} we get that $\Cycle$ either finds a cycle of length  at most $4g/3$ or 
there exists $w \in A_{1}$ such that $\delta(w, C) \le 2g/3$, as required.

Next, we assume the claim holds for $i-1$ and prove the claim for $i$.
Since the claim holds for $i-1$, either $\Cycle$ finds a cycle of length at most $4ig/3$ and since $4ig/3\le 4(i+1)g/3$ the claim holds, or there exists a vertex $w \in A_{i}$ such that $\delta(w, C) \le 2ig/3$. In this case it follows from  Lemma~\ref{L-Cycle-Correct-Small} that either $\Cycle$ finds a cycle of length at most $4(i+1)g/3$ or there exists $w' \in A_{i+1}$ such that $\delta(w', C) \le 2(i+1)g/3$, as required.
\end{proof}

Using Lemma~\ref{L-Cycle-Correct-Small-2}, it is straightforward to establish the correctness of $\Cycle$ in the case where $M(C) \leq g/3$.

\begin{corollary} \label{cor:cycle_correct_small}
Let $C$ be a cycle in $G$ such that $\ell(C)=g$ 
and suppose that $M(C)\leq g/3$.  
\mbox{\rm $\Cycle$} finds a cycle of length at most $\frac{4k}{3}g$.
\end{corollary}
\begin{proof}
When $i=k-1$ we have $A_k=\emptyset$ so there is no $w$ in $A_k$ and from Lemma~\ref{L-Cycle-Correct-Small-2} it follows that $\Cycle$ finds a cycle of length at most $\frac{4k}{3}g$.
\end{proof}

Next, we consider the case in which $M(C) > g/3$.
Let $(u,u')\in C$ such that $\ell(u,u')=M(C)$. 
We will show that if $\min\{d(u, A_{i}), d(u', A_{i})\}$ is relatively small then either $\Cycle$ finds a cycle of length at most $\frac{4(i+1)}{3}g$  or $\min\{d(u, A_{i+1}),d(u', A_{i+1})\}$  is relatively small. 

\begin{lemma} \label{L-Cycle-Correct-Big}
Let $C$ be a cycle in $G$ such that $\ell(C)=g$. Let $M(C)> g/3$,  
$(u, u')\in C$ and $\ell(u,u')=M(C)$. Let $0\leq i\leq k-1$.
If $\min\{d(u, A_{i}), d(u', A_{i})\} \le i \cdot (g-M(C))$
then either \mbox{\rm $\Cycle$} finds a cycle of length at most  $\frac{4(i+1)}{3}g$
or $\min\{d(u, A_{i+1}),d(u', A_{i+1})\} \le (i+1) \cdot (g-M(C))$. 
\end{lemma}


\begin{proof}
Let $0\leq i\leq k-1$ and $\min\{d(u, A_{i}), d(u', A_{i})\} \le i \cdot (g-M(C))$. 
If $\min\{d(u, A_{i+1}),d(u', A_{i+1})\} \le (i+1) \cdot (g-M(C))$ then the claim holds. We can assume, therefore, that $d(u', A_{i+1}) > (i+1) \cdot (g-M(C))$ and $d(u, A_{i+1}) > (i+1) \cdot (g-M(C))$. Additionally, we assume, without loss of generality, that $\delta(u, A_i) \le d(u', A_i)$. Since $\min\{d(u, A_{i}), d(u', A_{i})\} \le i \cdot (g-M(C))$, this implies that $\delta(u, A_i)=\delta(p_i(u),u)\le i(g-M(C))$. 

Let $r$ be the smallest number such that $\Cl(p_{i}(u))\cap G_r(p_i(u))$ contains a cycle. If $r\leq (i+1) \cdot (g-M(C))$ then by Lemma~\ref{L-shortest-path-cluster} $\ClusterOrCycle(p_{i}(u))$, when called, finds a cycle of length at most $2r\leq 2\cdot (i+1) \cdot (g-M(C))$. Since $g-M(C)\leq \frac{2}{3}g$ we have $2r\leq \frac{4(i+1)}{3}g$, and the claim holds.
Thus, we assume that $r> (i+1) \cdot (g-M(C))$. 

Next, we show that $u'\in \Cl_V(p_{i}(u))\cap V_{(i+1) \cdot (g-M(C))}(p_{i}(u))$. 
Since $\ell(C) = g$, $(u, u')\in C$ and $\ell(u,u')=M(C)$  we get that $\delta(u,u') = \min\{M(C), g-M(C)\}\le g-M(C)$. By the triangle inequality, $\delta(p_i(u), u') \leq \delta(u, A_i) + \ell(u,u')$. Combining these two inequalities with our assumption that $\delta(u, A_i) \leq i \cdot (g-M(C))$ yields
\begin{align*}
    \delta(p_i(u), u') &\LE \delta(u, A_i) + \delta(u,u') \\
    &\LE i \cdot (g-M(C)) + \delta(u,u') \\
    &\LE (i+1) \cdot (g-M(C)) 
\end{align*}
Thus, $u'\in V_{(i+1) \cdot (g-M(C))}(p_{i}(u))$.  Since $(i+1) \cdot (g-M(C)) <d(u', A_{i+1})$ we get that $\delta(p_i(u), u') <  d(u', A_{i+1})$ and thus $u'\in \Cl_V(p_{i}(u))$.
We conclude that $u'$ is in the graph $\Cl(p_{i}(u))\cap G_{(i+1) \cdot (g-M(C))}(p_{i}(u))$. 

Now since $r> (i+1) \cdot (g-M(C))$ it follows from Lemma~\ref{L-ClusterOrCycle} that $\ClusterOrCycle(p_{i}(u))$
computes $d(p_i(u),u')=\delta(p_i(u),u')$ and a shortest paths tree rooted at $p_i(u)$ that contains $u'$. 

Next, we show that when $\Cycle$ considers the edge $(u,u')$ it holds that $\pi(p_i(u), u)\neq (u',u)$ and $\pi(p_i(u), u')\neq (u,u')$.
We first show that $\pi(p_i(u), u)\neq (u',u)$. Assume for the sake of contradiction that $\pi(p_i(u), u) = (u',u)$.  This implies that  $\delta(p_i(u), u')< \delta(p_i(u), u)$. Since it always holds that $d(u', A_i)\leq \delta(p_i(u), u')$, we get that $d(u', A_i) < \delta(u, A_i)$, a contradiction to our assumption that $\delta(u, A_i) \le 
\delta(u', A_i)$.

We now show that $\pi(p_i(u), u')\neq (u,u')$.
Assume, for the sake of contradiction, that $\pi(p_i(u), u')= (u,u')$. This implies that 
$\delta(p_i(u),(u,u'))\leq \delta(p_i(u), u')\leq (i+1) \cdot (g-M(C))$, and hence 
$(u,u')$ is in $G_{(i+1) \cdot (g-M(C))}(p_{i}(u))$. 
Since $u'$ is in $\Cl(p_{i}(u))$ it follows from Lemma~\ref{L-shortest-path-cluster} that the shortest path between $p_i(u)$ and $u'$ is in $\Cl(p_{i}(u))$, thus its last edge $(u,u')$ is in $\Cl(p_{i}(u))$.
We conclude that $(u,u')$ is in $\Cl(p_{i}(u))\cap G_{(i+1) \cdot (g-M(C))}(p_{i}(u))$.

Consider a path $C'(u,u')$ between $u$ and $u'$ that uses the edges of $C\setminus \{(u,u')\}$. The length of this path is $g-M(C)$. 
Let  $P(p_i(u), u)$ be a shortest path between $p_i(u)$ and $u$. The length of this path is $\delta(u, A_i)\le i(g-M(C))$. 
The concatenation of $P(p_i(u), u)$ with $C'(u,u')$ is path between $p_i(u)$ and $u'$ of length at most $(i+1)(g-M(C))$ 
and thus the distance between $p_i(u)$ and each of the edges $C\setminus \{(u,u')\}$ is at most $(i+1)(g-M(C))$ which implies the edges of $C\setminus \{(u,u')\}$ are in $G_{(i+1) \cdot (g-M(C))}(p_{i}(u))$.

Let $(s,t)\in C'(u,u')$ and  assume that when going from $u$ to $u'$ on $C'(u,u')$ we first encounter $s$.  Let $C'(t,u')$ be the path from $t$ to $u'$ in $C$ avoiding the edge $(u,u')$.
Next we show that $(s,t)$ satisfies $\delta(p_i(u),s)+\ell(s,t)<d(t, A_{i+1})$, and thus in $\Cl(p_{i}(u))$. 

From the triangle inequality, we get that
$\delta(p_i(u),s)\leq \delta(p_i(u),u)+g-M(C)-\ell(s,t) - \ell(C'(t,u'))$. 
Thus, 
$\delta(p_i(u),s)+\ell(s,t)\leq  \delta(p_i(u),u)+g-M(C) - \ell(C'(t,u'))$.
Since $d(u', A_{i+1}) \leq \ell(C'(t,u')) + d(t, A_{i+1})$ we get that
$\delta(p_i(u),s)+\ell(s,t) \leq i(g-M(C)) + g-M(C)-\ell(C'(t,u'))< d(u', A_{i+1}) - \ell(C'(t,u')) \leq d(t, A_{i+1})$.

We conclude that $(s,t)$ is in $\Cl(p_{i}(u))$. 
Thus, there is a path between $p_i(u)$ and $u'$ in $\Cl(p_{i}(u))\cap G_{(i+1) \cdot (g-M(C))}(p_{i}(u))$ that does not use the edge $(u,u')$.

We reach a contradiction since there is a path between $p_i(u)$ and $u'$ 
in $\Cl(p_{i}(u))\cap G_{(i+1) \cdot (g-M(C))}(p_{i}(u))$ that does not use the edge $(u,u')$ and the edge $(u,u')$ is in  $\Cl(p_{i}(u))\cap G_{(i+1) \cdot (g-M(C))}(p_{i}(u))$, as well. 
However,  $\Cl(p_{i}(u))\cap G_{(i+1) \cdot (g-M(C))}(p_{i}(u))$ does not contain a cycle. 
We conclude that the condition in line~\ref{L-Cycle-If-pi} is true. 

Since $\delta(p_{i}(u), u')\le(i+1) \cdot (g-M(C))$, $\delta(p_{i}(u),u) = \delta(u, A_i) \le i \cdot (g-M(C))$, and $\ell(u, u')=M(C)$ we get that:
$$\delta(p_{i}(u),u) + d(p_{i}(u), u') + M(C)\le i \cdot (g-M(C)) + (i+1) \cdot (g-M(C)) + M(C) = 2i(g-M(C)) + g,$$
and algorithm $\Cycle$ finds a cycle of length at most $2i(g-M(C)) + g\leq \frac{4(i+1)}{3}g$, as required.
\end{proof}

Using Lemma~\ref{L-Cycle-Correct-Big} we show: 


\begin{lemma} \label{L-Cycle-Correct-Big-2}
Let $C$ be a cycle in $G$ such that $\ell(C)=g$. Let $M(C)> g/3$,  
$(u, u')\in C$ and $\ell(u,u')=M(C)$. Let $0\leq i\leq k-1$.
Either \mbox{\rm $\Cycle$} finds a cycle of length at most $\frac{4(i+1)}{3}g$  or $\min\{d(u, A_{i+1}),d(u', A_{i+1})\} \le (i+1) \cdot (g-M(C))$. 
\end{lemma}

\begin{proof}
We prove the claim by induction on $i$.
For the base case, we have $i=0$, thus $A_0=V$ and $C\cap A_0\neq \emptyset$. Let $z\in \{u,u'\}$. Since $\delta(z,z)=0$ we  use Lemma~\ref{L-Cycle-Correct-Big} and get that either $\Cycle$  finds a cycle of length at most
$\frac{4}{3}g$ or 
$\min\{d(u, A_{1}),d(u',A_1)\} \le (g-M(C))$.

Next, we assume the claim holds for $i-1$ and prove the claim for $i$.
Since the claim holds for $i-1$, then  
$\Cycle$ either finds a cycle of length at most $\frac{4i}{3}g$  and the claim holds, or $\min\{d(u, A_{i}),d(u', A_{i})\} \le i \cdot (g-M(C))$ and  we  use Lemma~\ref{L-Cycle-Correct-Big} and get that $\Cycle$ either finds a cycle of length at most
$\frac{4(i+1)}{3}g$ or $\min\{d(u, A_{i+1}),d(u', A_{i+1})\} \le (i+1) \cdot (g-M(C))$, as required.
\end{proof}

Using Lemma~\ref{L-Cycle-Correct-Big-2} it is straightforward to establish the correctness of $\Cycle$ for the case that $M(C) > g/3$.

\begin{corollary} \label{cor:cycle_correct_big}
Let $C$ be a cycle in $G$ such that $\ell(C)=g$ 
and let $M(C) > g/3$.  Let $(u, u')\in C$ and $\ell(u,u')=M(C)$.
\mbox{\rm $\Cycle$} finds a cycle of length at most  $\frac{4k}{3}g$. 
\end{corollary}

\begin{proof}
Since $A_k=\emptyset$, we get that $d(u, A_{k})=d(u',A_k)=\infty$, thus, it follows from Lemma~\ref{L-Cycle-Correct-Big-2} that 
$\Cycle$ finds a cycle of length at most  $\frac{4k}{3}g$. 
\end{proof}




\begin{lemma}\label{L-Cycle-Running-Time}
The expected running time of \mbox{\rm $\Cycle$} is $O((m+kn^{1+1/k})\log{n} + km)$.
\end{lemma}
\begin{proof}
From Lemma~\ref{L-Init} it follows that the call to $\Initialize$ takes $O((m+kn)\log n)$ time. 
For every $u\in V$ we call to  $\ClusterOrCycle(u)$. From Lemma~\ref{L-ClusterOrCycle} it follows that the running time of $\ClusterOrCycle(u)$ is ${O}(|\Cl_V(u)|\log n)$. From Lemma~\ref{L-cluster-size} it follows that $\EE[\sum_{u\in V}|\Cl_V(u)|]=O(kn^{1+1/k})$. Thus, 
calling  to $\ClusterOrCycle(u)$ for every $v\in V$ takes $O((kn^{1+1/k})\log{n})$ expected running time.

For every $(v,w)\in E$, we iterate over $k$ vertices. The cost of this is $O(km)$.  
\end{proof}

The proof of Theorem~\ref{thm:approx} follows from Corollary~\ref{cor:cycle_correct_small},  Corollary~\ref{cor:cycle_correct_big} and Lemma~\ref{L-Cycle-Running-Time}. 
\section{Lower bound for weighted approximation}
\label{sec:lower}

Here we prove a lower bound for girth computation (Theorem~\ref{thm:lower_bound}) under the following oracle model for accessing the edges of an $n=|V|$ vertex graph $G=(V, E,\ell)$.
Every vertex $v\in V$ has a counter, $c(v)$, initialized at $1$.
The following queries are allowed:
\begin{itemize}
\item For any $j\in \{1,\ldots,n\}$, access the $j$th vertex $v$ of $G$ and return its degree $\deg(v)$,
\item for any $j\in \{1,\ldots,n\}$, access the $j$th vertex $v$ and return the $c(v)$'th edge of $v$ in a predetermined sorted order in terms of non-decreasing edge weights; then increment $c(v)$. 
\end{itemize}

In other words, at any point, the algorithm can access the next weighted edge out of any vertex, so that to see the $i$th edge out of a vertex, the algorithm must have accessed all $i-1$ edges before it in the sorted order. For this model we prove Theorem~\ref{thm:lower_bound} restated below:

\qlower*

\newcommand{\davg}{d_\mathrm{avg}}

To prove Theorem~\ref{thm:lower_bound}, we provide a transformation from any unweighted graph of a given girth to a weighted graph where the girth has only increased and there are many vertices of large degree, where, if their largest incident edge length is sufficiently decreased, then the girth is decreased. The number of vertices and the degree in this transformation depend only on the average vertex degree in the original graph. Applying this transformation to a high-girth graph of large average degree, randomly decreasing the length of the longest edge incident to one of these high-degree vertices, and slightly perturbing the edge weights yields the distribution of graphs for our lower bound. 

In the rest of this section, we first provide the transformation in Lemma~\ref{lem:short_cycle} below and then we use it to prove Theorem~\ref{thm:lower_bound}. Lemma~\ref{lem:short_cycle} proves more properties about the transformation that are actually needed to prove the lower bound, but we include the proof as it is illustrated and of possible independent interest.

\begin{lemma}[Weighted Short-cycle Planting]
\label{lem:short_cycle}
For all $\epsilon \in [0, 1)$ if there exists an $n_0$-vertex $m_0$-edge unweighted graph of girth $g \in [3,\infty)$ then there exists a $n$-vertex $m$-edge weighted graph $G = (V,E,\ell)$ and vertex subset $S \subseteq V$ with the following properties:
\begin{itemize}
    \item \textbf{sizes}: $n \in [3 n_0, 4 n_0]$, $|S| \geq n_0$, and  $\ell_e \in [\epsilon, g]$ for all $e \in E$.
    \item \textbf{girth}: the girth of $G$ is at least $g$.
    \item \textbf{cycle planting}: each vertex in $S$ is incident to exactly one edge of length $\epsilon$, between $1$ and $2$ edges of length $g$, and between $\lfloor m_0 / (2n_0) \rfloor$ and $\lceil m_0 / n_0 \rceil$ edges of length $1$. Each length $g$ edge has both endpoints in $S$, and if it is changed to have length $1$, then the resulting graph has a cycle of length $1 + 2\epsilon$. 
\end{itemize}
\end{lemma}

\begin{proof}
Let $G_0 = (V_0,E_0,\ell_0)$ be an $n_0$-vertex, $m_0$-edge, unweighted graph ($\ell_0(e) = 1$ for all $e \in e$) of girth $g$. Further, for all $v \in V_0$, let $\deg(v)$ denote the degree of $v$ in $G_0$ and let $\davg \defeq \frac{2m_0}{n_0} = \frac{1}{n_0} \sum_{v \in V_0} \deg(v)$ denote the average degree of the vertices of $G_0$.

Given $G_0$, we construct $G = (V,E,\ell)$ and $S$ from it as follows. Informally, $G$ is the result of replacing every vertex $v$ in $G_0$ with a star connecting $\lceil \frac{2 \deg(v)}{\davg} \rceil$ vertices with edges of length $\epsilon$ and then dividing the edges of the original graph evenly over these new vertices (see Figure~\ref{fig:lbfig} below for a picture). 

\begin{figure}[ht]
\centering
\includegraphics{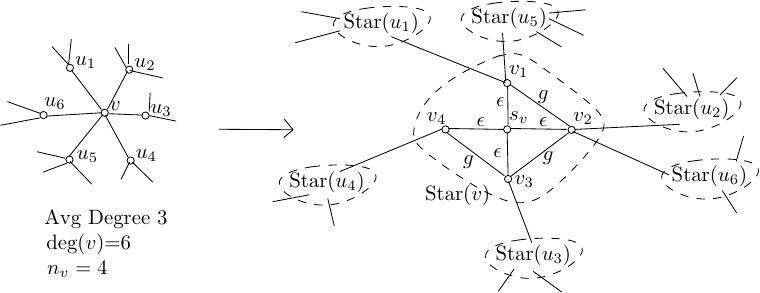}
\caption{An example of the lower bound construction. Here each vertex $v$ of $G_0$ gets replaced by a construction Star($v$) and each edge $(v,v')$ of $G_0$ is now between Star($v$) and Star($v'$) as shown.}
\label{fig:lbfig}
\end{figure}

Formally, to construct $G$ and $S$, we start with $G = G_0$ and $S = \emptyset$ and then apply the following procedure for every $v \in V_0$ one at a time. First, we create new vertices $v_1, \ldots, v_{n_v}$ for $n_v \defeq \lceil \frac{2 \deg(v)}{\davg} \rceil$ as well as a vertex $s_v$. Then for the edges $(v,u_1),\ldots,(v,u_{\deg(v)})$ incident to $v$ we replace them with the edges $(v_{(1 \,\mathrm{mod}\, n_v) + 1},u_1),\ldots,(v_{(\deg(v) \,\mathrm{mod}\, n_v) + 1},u_{\deg(v)}) \in E_0$ each of length $1$. After this, we delete the previous edges and $v$. Further, we add an edge $(s_v,v_i)$ of length $\epsilon$ for all $i \in [n_v]$. Finally, for all $i \in [n_v - 1]$ we add an edge of length $g$ between $v_i$ and $v_{i + 1}$ (note that if $n_v = 1$, then we add no such edges of length $g$) and if $n_v > 1$ then $v_1,\ldots,v_{n_v}$ is added to $S$. 

In the remainder of the proof, we show that $G$ and $S$ have the desired properties:

\paragraph{Sizes:} For every vertex $v \in V_0$ we add $n_v + 1$ vertices to $G$. Consequently, \[
n =\sum_{v \in V_0} (n_v + 1) = n_0 + \sum_{v \in V_0} \left\lceil \frac{2 \deg(v)}{\davg} \right\rceil
\,.
\]
Further, note that $\sum_{v \in V} n_v \in [2n_0, 3n_0]$ as $\sum_{v \in V_0} \frac{2 \cdot \deg(v)}{\davg} = \frac{4 m_0}{\davg} = 2n_0$. Consequently, $n \in [3 n_0,4 n_0]$ and $|S| \geq [\sum_{v \in V_0} n_v - 1 ] \geq 2n_0 - n_0 = n_0$. Finally $\ell_e \in [\epsilon,g]$ for all $e \in E$ by construction.

\paragraph{Girth:}
Let $C$ be a cycle in $G$ and let $u_1,\ldots,u_{c}$ denote the vertices of the cycle in order. Further, let $v^{(1)},\ldots,v^{(c)}$ denote their associated vertices in $V_0$ (i.e.\ $u_i = v^{(i)}_{j_i}$ or $s_{v^{(i)}}$ for all $i \in [c]$ for some $j_i$). 
By construction of $G$ for all $i \in [c]$ either $(v^{(i)},v^{((i\,\mathrm{mod}\,c)+1)}) \in E$ or $v^{(i)} = v^{((i\,\mathrm{mod}\,c)+1)}$. Consequently, if we simply remove vertex duplication in $v^{(1)},\ldots,v^{(c)}$ then the resulting vertex subsequence is either a cycle in $G_0$, in which case $\ell(C) \geq g$, or the sequence has exactly $1$ vertex, i.e.\ all the $v^{(i)}$ are the same (note that there cannot be only~$2$ distinct $v^{(i)}$ since an edge in $E_0$ would then have to be used twice). However, if all the $v^{(i)}$ are the same, then the cycle must be among the vertices $v_{1},\ldots,v_{n_v}$ and $s_v$ for some $v \in V_0$. Since the edges among these vertices of length $\epsilon$ are acyclic, the cycle must use one of the edges of length $g$, and therefore $\ell(C) \geq g$. 

\paragraph{Cycle planting:}
Note that by construction every vertex in $S$ is $v_i$ for some $v \in V_0$ and $i \in [n_v - 1]$. Further, by construction, $v_i$ is incident to one edge of length $\epsilon$, between $1$ and $2$ edges of length $g$, and some number of length $1$. Further, if the edge of length $g$ is given length $1$ then $v_i,v_{i+1},s_v$ yields a cycle of length $1 + 2\epsilon$. Consequently, it only remains to bound the number  of edges incident to $v_i$ of length $1$. However, the number of such edges is either $\lfloor \deg(v) / n_v \rfloor$ or $\lceil \deg(v) / n_v \rceil$. Further, since $n_v > 1$ (as $v_i \in S$) this implies that $1 < \frac{2 \deg(v)}{\davg}$ and the result then follows as
\begin{align*}
\left\lfloor \frac{\deg(v)}{n_v} \right\rfloor
&\geq
\left\lfloor \frac{\deg(v)}{\left(\frac{2 \deg(v)}{\davg}\right) + 1} \right\rfloor
>
\left\lfloor \frac{\deg(v)}{\left(\frac{4 \deg(v)}{\davg}\right)} \right\rfloor
=
\left\lfloor \frac{m_0}{2n_0} \right\rfloor
\text{ and }
\\
\left\lceil \frac{\deg(v)}{n_v} \right\rceil
&\leq 
\left\lceil \frac{\deg(v)}{\left(\frac{2 \deg(v)}{\davg}\right)} \right\rceil
=
\left\lceil \frac{\davg}{2} \right\rceil
=
\left\lceil \frac{m_0}{n_0} \right\rceil
\,.
\end{align*}
\end{proof}


\begin{proof}[Proof of Theorem~\ref{thm:lower_bound}]
Let $k\geq 1$ be a fixed integer, and let $\tau \in (0,1)$ be given.
Let $G_k$ be a $n_0$ node girth conjecture graph for $k$, i.e., $G_0$ has 
 $m_0=\Theta(n_0^{1+1/k})$ edges and girth $2k+2$. Apply Lemma~\ref{lem:short_cycle} to $G_k$ and any $\epsilon<\tau/(2(2k-2-\tau))$, obtaining a graph $G$.
We know that no matter which edge of weight $g$ we pick, if we change its weight to $1$, the girth goes from $\geq g$ to $1+2\epsilon$.
Recall also that $G$ contains a vertex set $S \subseteq V$ with $|S| \geq n_0$ such that the number of edges incident to each vertex of $S$ is between $2+\lfloor m_0/(2n_0)\rfloor$ and $3+\lceil m_0/n_0\rceil$, at most $2$ of which are of weight $>1$. 

Leveraging $G_k$, $G$, and the properties of $G$ given Lemma~\ref{lem:short_cycle} we prove our 
lower bound below:

Consider any deterministic algorithm $A$ for girth approximation running on $G$.
Suppose that $A$ accesses $< \frac{n_0}{4}\lfloor m_0/(2n_0)\rfloor$ edges and let $S_q$ denote the subset of $S$ consisting of vertices which the algorithm queried at least $\lfloor m_0/(2n_0)\rfloor$ times. Note that $|S_q| < \frac{n_0}{4}$. Further, every vertex in $S_q$ is incident to at least $\lfloor m_0/(2n_0)\rfloor$ edges of length $1$ and at most $2$ edges of length $g > 1$, and all edges of length $g$ have both endpoints in $S$. Consequently, at most $2|S_q| < \frac{n_0}{2}$ edges of length $g$ are accessed via queries. However, there are at least $|S|/2 \geq n_0/2$ edges of length $g$ in the graph. Therefore, at least one of the edges of length $g$ is not accessed, and $A$ would perform the same when run on $G$ and when run on $G$ for which one of the weight $g$ edges incident to $s$ were changed to have any length $\geq 1$.

Thus $A$ will fail to distinguish between girth $1+2\epsilon$ and girth $\geq 2k+2$. Further, if the algorithm outputs a cycle containing the edge of length $g$ that was not accessed, then its length could be changed to be arbitrarily large, and the accesses would be consistent. On the other hand, if the algorithm does not output a cycle containing this edge of length $g$ its length could be changed to have length $1,$ and the cycle will have length $\geq 2k + 2$ although the girth is $1 + 2\epsilon$. Consequently, in the worst case the ratio of the girth to the length of the cycle output is at least $(2k + 2)/(1 + 2 \epsilon)$ Since $\epsilon<\tau/(2(2k-2-\tau))$, we have that $-\tau(1+2\epsilon)+2\epsilon(2k-2)<0$ and $(2k-2-\tau)(1+2\epsilon)<(2k-2)$.

Thus any deterministic algorithm that makes fewer than $\frac{n_0}{4}\lfloor m_0/(2n_0)\rfloor = \Theta(n_0^{1+1/k})$ queries, will not be able to compute a cycle $C$ with $\ell(C) \leq (2k+2-\tau)g$ on some input.
\end{proof}

\bibliographystyle{alpha}
\bibliography{bibliography}
\end{document}